\newcommand{\CLASSINPUTtoptextmargin}{2cm}
\documentclass[conference]{IEEEtran}
\IEEEoverridecommandlockouts
\usepackage{cite}
\usepackage{amsmath,amssymb,amsfonts,amsthm,bm,mathtools}
\usepackage{algorithm}
\usepackage{algpseudocode}
\usepackage{graphicx}
\usepackage{textcomp}
\usepackage{xcolor}
\def\BibTeX{{\rm B\kern-.05em{\sc i\kern-.025em b}\kern-.08em
    T\kern-.1667em\lower.7ex\hbox{E}\kern-.125emX}}

\newcommand{\C}{\mathbb{C}}
\newcommand{\E}{\mathbb{E}}
\newcommand{\tr}{\operatorname{tr}}
\newcommand{\diag}{\operatorname{diag}}

\newcommand{\Herm}{\mathsf{H}}
\newcommand{\T}{\mathsf{T}}
\newcommand{\I}{\mathbf{I}}
\newcommand{\0}{\mathbf{0}}

\newcommand{\snrdl}{\mathsf{snr}_{\mathrm{dl}}}

\newcommand{\Ch}{\mathbf{C}_{\mathbf{h}_k}}
\newcommand{\Cu}{\mathbf{C}_{\mathbf{\tilde h}_k }}
\newcommand{\Cub}{\mathbf{C}_{\mathbf{\tilde h}_k }^{(b)}}
\newcommand{\Cue}{\mathbf{C}_{\mathbf{\tilde h}_k }^{(e)}}
\newcommand{\Cq}{\mathbf{C}_{\mathbf{q}}}

\newcommand{\Kmat}{\mathbf{K}}
\newcommand{\Fmat}{\mathbf{F}}
\newcommand{\Xmat}{\mathbf{X}}

\newlength{\figgap}             % 1. 변수 'figgap' 생성 (이름 맘대로 가능)
\definecolor{royalblue}{RGB}{0, 80, 200}
 
\newtheorem{theorem}{Theorem}
\newtheorem{Theorem 2}{Theorem 2}

\theoremstyle{remark}

\begin{document}

\title{CSI Compression for Massive MIMO-OFDM: Mismatch-Aware Rate-Distortion Trade-offs}
% {\footnotesize \textsuperscript{*}Note: Sub-titles are not captured for https://ieeexplore.ieee.org  and

\author{\IEEEauthorblockN{Bumsu Park, Youngmok Park, Chanho Park, and Namyoon Lee}
\IEEEauthorblockA{\textit{Department of Electrical Engineering,} \\
POSTECH,\\
Pohang, South Korea \\
Email: \{bumsupark, ympark1999, chanho26, nylee\}@postech.ac.kr}
}

\maketitle

\begin{abstract}
 We study channel state information (CSI) compression for wideband frequency division duplex massive multiple-input multiple-output (MIMO) when the base station (BS) reconstructs CSI using an imperfect covariance model. Under matched second-order statistics, remote rate--distortion theory yields transform coding with reverse water-filling (RWF) over covariance eigenmodes. With decoder-side covariance mismatch, however, this allocation is no longer end-to-end optimal. We derive an achievable mismatched Gaussian rate--distortion characterization based on a Gaussian test channel and a mismatched minimum mean square error (MMSE) reconstruction rule. In a shared-eigenvector regime (common eigenbasis, mismatched eigenvalues), the problem decouples across modes and leads to a robust reverse water-filling (RRWF) allocation computable via bisection and per-mode root finding. Simulations using  wideband massive MIMO covariance models show that RRWF consistently improves reconstruction distortion and end-to-end mean square error relative to conventional RWF under mismatch.\end{abstract}

% \begin{IEEEkeywords}
% component, formatting, style, styling, insert.
% \end{IEEEkeywords}

  \section{Introduction}\label{sec:intro}
Massive multiple-input multiple-output (MIMO) turns many antennas into spatial degrees of freedom, but the gains hinge on accurate channel state information at the transmitter (CSIT) \cite{marzetta2010noncooperative,5GMag}. In frequency-division duplex (FDD), CSIT must be acquired from downlink (DL) training and fed back from the user equipment (UE) \cite{Lee2014,Caire2006,kim2024fdd,Han2024}. The dimension of this feedback problem scales with both array size and bandwidth, e.g., orthogonal frequency division multiplexing (OFDM) subcarriers, and becomes a bottleneck in wideband massive MIMO.

Recently, learning-based channel state information (CSI) feedback has shown impressive empirical gains by training autoencoder-style architectures end-to-end \cite{deeplearning2018,Guo2022_Overview,Lightweight_CSI_feedback_2021,Attention_CSI_Feedback_2022,Learningtopitmize_2025,park2025transformer,park2024multi,wen2018deep}. Starting from CsiNet \cite{deeplearning2018}, subsequent designs have added deeper convolutional backbones and attention/transformer modules to improve reconstruction accuracy \cite{Lightweight_CSI_feedback_2021,Attention_CSI_Feedback_2022,park2025transformer}. The cost, however, is paid where it hurts most: at the UE. State-of-the-art transformer encoders are memory- and compute-intensive—often an order of magnitude heavier than earlier networks—making them difficult to deploy on power- and silicon-constrained devices. This brings the CSI feedback problem to a familiar engineering tension: how to retain the performance of sophisticated models while respecting the tight complexity budget at the edge.

A key observation is that, unlike images, massive MIMO channels admit a strong statistical structure. Over a wide range of measurements and standardized models, the DL channel is well-approximated as a correlated complex Gaussian vector: each antenna sees a superposition of many scattering components, and as the array grows, non-Gaussian effects average out, leaving the channel largely characterized by its covariance \cite{COST2100,3GPP38901,AdhikaryJSDM,NYUSIM,Han2024}. This covariance encodes the underlying geometry—angular spreads, clusters, and frequency selectivity—and provides a natural basis for compression.

A useful benchmark comes from remote rate--distortion (RD) theory. Under the standard correlated complex Gaussian channel model, the UE forms a minimum mean square error (MMSE) estimate from pilots, then compresses that estimate under a feedback-rate constraint, and the base station (BS) reconstructs it \cite{Caire2025,CoverThomas}. When encoder and decoder share the same second-order statistics, the MMSE estimate is Gaussian and the optimal RD trade-off is achieved by transform coding with reverse water-filling (RWF) over the covariance eigenmodes \cite{TC2000TIT,TC2001Mag,Caire2025}. This establishes both a performance limit and a concrete design: a covariance-matched linear transform and a mode-dependent bit allocation.

The matched-statistics assumption, however, is fragile in FDD. The UE can infer the relevant covariance from its DL observation model, but the BS often reconstructs using an \emph{imperfect} covariance due to uplink--downlink calibration errors, nonstationarity, and finite-sample covariance learning noise \cite{khalilsarai2018fdd,almradi2020hybrid,bjornson2014massive,yang2023structured}. Under such mismatch, classical RWF need not be end-to-end optimal: bits can be spent on modes that the decoder cannot reliably exploit, while modes that matter under the decoder’s model may be under-served.

This paper develops a mismatch-aware RD benchmark for FDD CSIT feedback. Our goal is to quantify the fundamental penalty of decoder-side covariance mismatch and to derive a tractable, robust bit-allocation rule. Concretely:
\begin{itemize}
\item We derive an \emph{achievable} mismatched Gaussian RD characterization for compressing the UE-side MMSE estimate using a Gaussian test channel and a mismatched MMSE reconstruction rule at the BS.
\item We specialize to a shared-eigenvector regime (common eigenbasis, mismatched eigenvalues), motivated by the relative stability of long-term spatial subspaces. In this case the problem decouples into parallel scalar modes, and the optimal allocation is characterized by a per-mode Karush–Kuhn–Tucker (KKT) condition, yielding a \emph{robust reverse water-filling}
(RRWF) rule. 

\item We provide an implementable procedure based on bisection over the Lagrange multiplier and per-mode root finding to compute the RRWF allocation efficiently.
\end{itemize}
Numerical results for massive MIMO-OFDM covariance models show that RRWF consistently reduces reconstruction distortion and improves end-to-end normalized mean square error (NMSE) relative to classical RWF when the BS reconstruction statistics are mismatched.

\section{CSI Compression and Reconstruction Problem}\label{sec:system}
We consider a single cell massive MIMO-OFDM systems, where a BS with $M$ antennas operates in OFDM mode over $N$ subcarriers and serves $K\le M$ single-antenna users. The channel is block fading, i.e., CSI is constant over one coherence frame and changes independently frame to frame. During DL common training, the BS transmits pilots over $T_p$ out of $T$ OFDM symbols and over a subset $\mathcal{N}_p\subseteq\{1,\dots,N\}$ with $N_p\triangleq |\mathcal{N}_p|\le N$ subcarriers, yielding pilot dimension
\begin{equation}
L^{\rm tr} = T_p N_p.
\end{equation}

\subsection{Downlink Pilot Observation Model}

Let $\mathbf{h}_k[n]\in\C^M$ be the DL channel vector of user $k$ on subcarrier $n$ and let $\mathbf{X}[n]\in\C^{T_p\times M}$ be the pilot matrix on subcarrier $n$. The received pilot signal of user $k$ over subcarrier $n\in\mathcal{N}_p$ is
\begin{equation}\label{eq:pilot_scalar}
\mathbf{y}^{\mathrm{tr}}_{k}[n] = \mathbf{X}[n]\mathbf{h}_k[n] + \mathbf{n}^{\mathrm{tr}}_{k}[n],
\quad \forall k\in\{1,\dots,K\},~ n\in\mathcal{N}_p,
\end{equation}
where $\mathbf{n}^{\mathrm{tr}}_{k}[n]\sim\mathcal{CN}(\0,\I)$ is AWGN. The pilot entries satisfy
\begin{equation}
X_{i,j}[n] \sim \mathcal{CN}\!\left(0,\frac{\snrdl}{M}\right),
\end{equation}
so that $\snrdl$ corresponds to a DL pre-beamforming signal-to-noise ratio (SNR) under normalized noise power.

Stacking all $L^{\rm tr}$ pilot observations, define
\begin{equation}
\mathbf{y}^{\mathrm{tr}}_{k}
\triangleq
\big[\mathbf{y}^{\mathrm{tr}}_{k}[n_1]^{\T},\dots,\mathbf{y}^{\mathrm{tr}}_{k}[n_{N_p}]^{\T}\big]^{\T}
\in\C^{L^{\rm tr}},
\end{equation}
and the stacked channel
\begin{equation}
\mathbf{h}_k
\triangleq
\big[\mathbf{h}_k[1]^{\T},\dots,\mathbf{h}_k[N]^{\T}\big]^{\T}
\in\C^{MN}.
\end{equation}
From the statitical model 
the stacked channel is accurately modeled as the complex Gaussiaon with covariance matrix $\Ch=\E[\mathbf{h}_k\mathbf{h}_k^{\Herm}]$, i.e., 
\begin{align}
    \mathbf{h}_k\sim\mathcal{CN}(\0,\Ch).
\end{align}
 Then the overall pilot observation is
\begin{equation}\label{eq:pilot_stacked}
\mathbf{y}^{\mathrm{tr}}_{k} = \Xmat \mathbf{h}_k + \mathbf{n}^{\mathrm{tr}}_{k}, \quad \mathbf{n}^{\mathrm{tr}}_{k}\sim\mathcal{CN}(\0,\I_{L^{\rm tr}}),
\end{equation}
where $\Xmat\in\C^{L^{\rm tr}\times MN}$ collects all pilot matrices across probed subcarriers (with zero blocks for unprobed subcarriers).

\subsection{ CSI Estimation}
Using the pilot matrix ${\bf X}$, the user forms the MMSE estimate
\begin{equation}\label{eq:mmse}
\mathbf{\tilde h}_k \triangleq \E[\mathbf{h}_k|\mathbf{y}^{\mathrm{tr}}_{k}]
= \Ch \Xmat^{\Herm}\big(\Xmat\Ch\Xmat^{\Herm} + \I_{L^{\rm tr}}\big)^{-1}\mathbf{y}^{\mathrm{tr}}_{k},
\end{equation}
whose covariance is
\begin{equation}\label{eq:Cu}
\Cu = \E[\mathbf{\tilde h}_k \mathbf{\tilde h}_k ^{\Herm}]
= \Ch \Xmat^{\Herm}\big(\Xmat\Ch\Xmat^{\Herm} + \I_{L^{\rm tr}}\big)^{-1}\Xmat\Ch,
\end{equation}
and the corresponding mean square error (MSE) boils down to
\begin{align}\label{eq:Dmmse}
D_k^{\mathrm{mmse}} =\tr(\Ch-\Ch \Xmat^{\Herm}\big(\Xmat\Ch\Xmat^{\Herm} + \I_{L^{\rm tr}}\big)^{-1}\Xmat\Ch).
\end{align}
The resulting MMSE channel estimate \(\mathbf{\tilde h}_k\) follows a circularly
symmetric complex Gaussian distribution with covariance matrix
\(\mathbf{C}_{\mathbf{\tilde h}_k}\) defined in \eqref{eq:Cu}, i.e.,
\begin{equation}
\mathbf{\tilde h}_k \sim \mathcal{CN}\!\left(\mathbf{0},
\mathbf{C}_{\mathbf{\tilde h}_k}\right).
\end{equation}
Since \(\mathbf{\tilde h}_k\) is a complex Gaussian vector, it can be optimally compressed using a transform coding framework based on the eigen-decomposition of \(\mathbf{C}_{\mathbf{\tilde h}_k}\), followed by RWF rate allocation across the resulting eigenmodes \cite{TC2000TIT,Caire2025}. This approach is known to achieve the information-theoretic RD limit for complex Gaussian sources under mean-squared error distortion, as established in the classical literature \cite{CoverThomas,TC2000TIT,TC2001Mag}.

\subsection{Noisy CSI Compression under Encoder and Decoder Knowledge Mismatch}
In practice, the UE and the BS rarely share identical second-order statistics for CSI compression and reconstruction \cite{khalilsarai2018fdd, yang2023structured, almradi2020hybrid, bjornson2014massive}. Such covariance mismatch naturally arises due to several factors, including imperfect uplink–downlink calibration that breaks channel reciprocity, differences in training length and pilot power between uplink and downlink, and estimation errors in covariance learning. As a result, the encoder and decoder often rely on mismatched covariance information, which fundamentally alters the optimal CSI compression and reconstruction strategy.

We model this asymmetry by allowing the encoder and decoder to rely on different covariance matrices:
\begin{itemize}
\item the \emph{encoder-side} covariance model $\Cue \succ \mathbf{0}$ available at the UE, and
\item the \emph{decoder-side} covariance model $\Cub \succ \mathbf{0}$ available at the BS.
\end{itemize}
The true source law is governed by $\Cu$, while $\Cue$ and $\Cub$ represent (possibly outdated or imperfect) covariance knowledge used for designing the encoder and decoder, respectively.

We consider a blocklength-$n$ source coding formulation for compressing an i.i.d. sequence 
$\{\mathbf{\tilde h}_k[t]\}_{t=1}^n$.
An $(n,R)$ compression scheme consists of
\begin{align}
&\text{Encoder:}\quad 
\phi_n: (\C^{MN})^n \times \mathbb{S}_{++}^{MN} \to \{1,\dots,2^{nR}\}, \nonumber\\
&\hspace{22mm} {\bf b}_k = \phi_n(\mathbf{\tilde h}_k^n;\Cue), \label{eq:enc_mismatch}\\[1mm]
&\text{Decoder:}\quad 
\psi_n: \{1,\dots,2^{nR}\}\times \mathbb{S}_{++}^{MN} \to (\C^{MN})^n, \nonumber\\
&\hspace{22mm} {\mathbf{\hat u}}_k^n = \psi_n({\bf b}_k;\Cub), \label{eq:dec_mismatch}
\end{align}
where $\mathbb{S}_{++}^{MN}$ denotes the cone of $MN\times MN$ Hermitian positive definite matrices.
The encoder is designed using $\Cue$, while the decoder is designed using $\Cub$.
 
The performance metric is the \emph{true} MSE under the true distribution $\mathbf{\tilde h}_k\sim\mathcal{CN}(\mathbf{0},\Cu)$:
\begin{equation}\label{eq:dist_mismatch_block}
\mathsf{D}_n(R;\Cu,\Cue,\Cub)
\triangleq
\frac{1}{n}\E\Big[\big\|\mathbf{\tilde h}_k^n-{\mathbf{\hat u}}_k^n\big\|^2\Big],
\end{equation}
where the expectation is taken w.r.t. the true law of $\mathbf{\tilde h}_k^n$ and the (possibly randomized) encoder/decoder.

We define the mismatch-aware distortion--rate (DR) function as
\begin{align}\label{eq:DR_mismatch_def}
&D_k^{\mathrm{quant}}(R;\Cu,\Cue,\Cub)
\nonumber\\
&\triangleq
\inf_{\{\phi_n,\psi_n\}_{n\ge 1}}
\ \limsup_{n\to\infty}\ \mathsf{D}_n(R;\Cu,\Cue,\Cub),
\end{align}
where the infimum is over all sequences of encoders $\phi_n(\cdot;\Cue)$ and decoders $\psi_n(\cdot;\Cub)$ satisfying the rate constraint $|{\bf b}_k|\le 2^{nR}$.
 
Combining the standard remote DR decomposition and the covariance mismatch effect~\cite{Caire2025}, we define the end-to-end CSI distortion as
\begin{equation}\label{eq:remote_mismatch_benchmark}
D_{k}^{\mathrm{E2E}}(R)
=D_k^{\mathrm{mmse}}
+
 D_k^{\mathrm{quant}}(R;\Cu,\Cue,\Cub),
\end{equation}
where $D_k^{\mathrm{mmse}}=\tr(\Ch-\Cu)$ is the MMSE error from DL training derived in \eqref{eq:Dmmse}.

Without loss of generality, to isolate the impact of encoder--decoder statistical mismatch, we assume that the encoder has perfect knowledge of the source covariance, i.e., $\Cue=\Cu$, while the decoder reconstructs using an imperfect covariance model $\Cub\neq\Cu$ in the sequel.

\section{Mismatch-Aware Robust Rate--Distortion  }\label{sec:theorem}

We provide an achievable mismatch-aware RD characterization for compressing $\mathbf{\tilde h}_k $ under a Gaussian test channel and mismatched MMSE reconstruction based on $\Cub$.

The following theorem is our major result of this paper. 

\begin{theorem}[Mismatched Rate-Distortion]\label{thm:mismatchRD}
Fix any $\Cq\succ\0$ and consider the Gaussian test channel
\begin{equation}\label{eq:testchannel}
\mathbf{\hat h}_k = \mathbf{\tilde h}_k +\mathbf{q}, \qquad
\mathbf{q}\sim\mathcal{CN}(\0,\Cq), \qquad \mathbf{q}\perp\!\!\!\perp \mathbf{\tilde h}_k .
\end{equation}
Let the BS reconstruct using the \emph{mismatched} MMSE filter computed from $\Cub$:
\begin{equation}\label{eq:mismatchedMMSE}
\widetilde{\mathbf{u}}_k = \Kmat \mathbf{\hat h}_k, \qquad
\Kmat \triangleq \Cub(\Cub+\Cq)^{-1}.
\end{equation}
Then there exists a sequence of (block) source codes such that the achievable rate satisfies
\begin{equation}\label{eq:rate_det}
R(\Cq) = \log_2\frac{|\Cu+\Cq|}{|\Cq|}
= \log_2\big|\I+\Cu \Cq^{-1}\big|,
\end{equation}
and the resulting end-to-end distortion (true MSE) is
\begin{align}\label{eq:dist_trace}
&D_k^{\mathrm{quant}}(\Cq)
= \E\big[\|\mathbf{\tilde h}_k -\widetilde{\mathbf{u}}_k\|^2\big] \nonumber\\
%&= \tr\!\Big(
%(\I-\Kmat)\Cu(\I-\Kmat)^{\Herm} + \Kmat \Cq \Kmat^{\Herm}
%\Big) \nonumber\\
&= \tr\!\Big(
\Cu - \Kmat\Cu - \Cu \Kmat^{\Herm} + \Kmat(\Cu+\Cq)\Kmat^{\Herm}
\Big),
\end{align}
where $\Kmat=\Cub(\Cub+\Cq)^{-1}$.

Consequently, for any target rate $R$, the following mismatch-aware achievable distortion bound holds:
\begin{equation}\label{eq:RD_ach}
D_k^{\mathrm{quant}}(R)\ \le\
\min_{\Cq\succ \0:\ R(\Cq)\le R}\ D_k^{\mathrm{quant}}(\Cq).
\end{equation}
Moreover, if $\Cub=\Cu$ and $\Cq$ is optimized in the eigenbasis of $\Cu$, \eqref{eq:RD_ach} reduces to the classical Gaussian RWF solution.
\end{theorem}

\begin{proof}
\label{subsec:proof_thm}
The proof of Theorem~\ref{thm:mismatchRD}
 consists of (i) achievability of the rate expression \eqref{eq:rate_det} for the Gaussian test channel, and (ii) evaluation of the true distortion under mismatched MMSE reconstruction \eqref{eq:mismatchedMMSE}.
 
 Under \eqref{eq:testchannel}, $\mathbf{\hat h}_k\sim\mathcal{CN}(\0,\Cu+\Cq)$ and
\begin{equation}\label{eq:mutualinfo}
\mathcal{I}(\mathbf{\tilde h}_k ;\mathbf{\hat h}_k)
= h(\mathbf{\hat h}_k)-h(\mathbf{\hat h}_k|\mathbf{\tilde h}_k )
= h(\mathbf{\hat h}_k)-h(\mathbf{q}),
\end{equation}
where $\mathcal{I}(X;Y)$ and $h(X)$ denote the mutual information between $X$ and $Y$ and the differential entropy of $X$, respectively. For an $n$-dimensional circularly symmetric complex Gaussian vector $\mathbf{z}\sim\mathcal{CN}(\0,\mathbf{C})$, the differential entropy is
\[
h(\mathbf{z}) = \log\big((\pi e)^n|\mathbf{C}|\big).
\]
Applying this to \eqref{eq:mutualinfo} yields
\[
R(\Cq)= \mathcal{I}(\mathbf{\tilde h}_k ;\mathbf{\hat h}_k)
=\log_2\frac{|\Cu+\Cq|}{|\Cq|},
\]
which is \eqref{eq:rate_det}.
By standard random coding arguments for memoryless Gaussian sources, there exists a sequence of vector quantizers (source codes) approaching the test channel \eqref{eq:testchannel} with rate arbitrarily close to $R(\Cq)$.

Given $\mathbf{\hat h}_k=\mathbf{\tilde h}_k +\mathbf{q}$ and $\widetilde{\mathbf{u}}_k=\Kmat\mathbf{\hat h}_k$ with $\Kmat=\Cub(\Cub+\Cq)^{-1}$, the error is
\begin{equation}\label{eq:error}
\mathbf{e}\triangleq \mathbf{\tilde h}_k -\widetilde{\mathbf{u}}_k
= (\I-\Kmat)\mathbf{\tilde h}_k - \Kmat\mathbf{q}.
\end{equation}
Since $\mathbf{\tilde h}_k $ and $\mathbf{q}$ are independent and zero mean,
\begin{align}\label{eq:err_cov}
\E[\mathbf{e}\mathbf{e}^{\Herm}]
&= (\I-\Kmat)\E[\mathbf{\tilde h}_k \mathbf{\tilde h}_k ^{\Herm}](\I-\Kmat)^{\Herm}
+ \Kmat \E[\mathbf{q}\mathbf{q}^{\Herm}] \Kmat^{\Herm}\nonumber\\
&= (\I-\Kmat)\Cu(\I-\Kmat)^{\Herm}+\Kmat \Cq \Kmat^{\Herm}.
\end{align}
Therefore,
\[
D_k^{\mathrm{quant}}(\Cq)=\E[\|\mathbf{e}\|^2]=\tr(\E[\mathbf{e}\mathbf{e}^{\Herm}]),
\]
which gives the first line of \eqref{eq:dist_trace}. Expanding
\[
(\I-\Kmat)\Cu(\I-\Kmat)^{\Herm} = \Cu - \Kmat\Cu - \Cu \Kmat^{\Herm} + \Kmat\Cu\Kmat^{\Herm}
\]
and adding $\Kmat\Cq\Kmat^{\Herm}$ yields the second line of \eqref{eq:dist_trace}.

Each $\Cq\succ 0$ yields an achievable pair $(R(\Cq),D_k^{\mathrm{quant}}(\Cq))$. Minimizing distortion subject to $R(\Cq)\le R$ yields \eqref{eq:RD_ach}.
When $\Cub=\Cu$, optimizing $\Cq$ in the eigenbasis of $\Cu$ recovers the classical Gaussian RD RWF \cite{TC2000TIT,CoverThomas}.
\end{proof}

Theorem 1 provides a RD trade-off under a general structure of the encoder and decoder structure mismatch. Therefore, this result can be  extend to various types of channel model (e.g, Gaussian mixture  \cite{park2026fundamental,Benedikt2022GMM}). However, it is intractable to optimize $\Cq$ to minimize distortion subject to $R(\Cq)\le R$ due to the general strucutre of $D_k^{\mathrm{quant}}(\Cq)$. 

\section{Mismatch-Aware Robust Reverse Waterfilling}\label{sec:Theorem 2}
A tractable and practically relevant regime is when $\Cu$ and $\Cub$ share eigenvectors (e.g., when $M\to \infty$ for array antennas, eigenvectors approximate deterministic DFT vectors \cite{AdhikaryJSDM}, by Szego's asymptotic results on Toeplitz matrices \cite{grenander1958toeplitz}). We restrict $\Cq$ to the same eigenbasis. The following theorem shows how to generalize the conventional RWF rate allocation strategy when the decoder has imperfect knowledge of the eigenvalues. 

\begin{theorem}[Rate-distortion under eigenvlaue mismatch]\label{cor:modifiedWF}
Assume there exists a unitary $\Fmat$ such that
\begin{equation}\label{eq:simdiag}
\Cu = \Fmat \diag(\lambda_1,\dots,\lambda_{MN})\Fmat^{\Herm}
\end{equation}
and
\begin{equation}
\Cub = \Fmat \diag(\lambda^{(b)}_1,\dots,\lambda^{(b)}_{MN})\Fmat^{\Herm},
\end{equation}
with $\lambda_i>0$ and $\lambda_i^{(b)}>0$. Restrict
\begin{equation}\label{eq:Cq_diag}
\Cq = \Fmat \diag(d_1,\dots,d_{MN})\Fmat^{\Herm},\qquad d_i>0.
\end{equation}
Then the achievable rate decomposes as
\begin{equation}\label{eq:rate_sum}
R = \sum_{i=1}^{MN} r_i(d_i),\qquad
r_i(d_i)=\log_2\Big(1+\frac{\lambda_i}{d_i}\Big),
\end{equation}
and the true distortion decomposes as
\begin{equation}\label{eq:D_sum}
D_k^{\mathrm{quant}} = \sum_{i=1}^{MN} e_i(d_i),
\qquad
e_i(d_i)=\frac{\lambda_i d_i^2 + (\lambda^{(b)}_i)^2 d_i}{(\lambda^{(b)}_i+d_i)^2}.
\end{equation}
Moreover, the optimal allocation $\{d_i^\star\}$ solving
\begin{equation}\label{eq:opt_problem}
\min_{d_i>0}\ \sum_{i=1}^{MN} e_i(d_i)
\quad \text{s.t.}\quad
\sum_{i=1}^{MN} r_i(d_i)\le R
\end{equation}
satisfies: there exists $\mu\ge 0$ such that for each active mode $i$,
\begin{equation}\label{eq:cubic}
(\ln 2)\,\lambda_i^{(b)}\, d_i(\lambda_i+d_i)\Big((\lambda_i^{(b)})^2+(2\lambda_i-\lambda_i^{(b)})d_i\Big)
=
\mu\,\lambda_i(\lambda_i^{(b)}+d_i)^3.
\end{equation}
When $\lambda_i^{(b)}=\lambda_i$ for all $i$ (matched case), \eqref{eq:cubic} reduces to the classical RWF solution.
\end{theorem}

\begin{proof}
Let $\mathbf{z}\triangleq \Fmat^{\Herm}\mathbf{\tilde h}_k $ and $\widehat{\mathbf{z}}\triangleq \Fmat^{\Herm}\mathbf{\hat h}_k$.
Because $\Fmat$ is unitary and \eqref{eq:simdiag}--\eqref{eq:Cq_diag} hold, the test channel decouples into independent scalar modes:
\[
z_i \sim \mathcal{CN}(0,\lambda_i),\qquad
\widehat{z}_i=z_i+q_i,\quad q_i\sim\mathcal{CN}(0,d_i),
\]
with $q_i\perp\!\!\!\perp z_i$.

For each mode,
$
\mathcal{I}(z_i;\widehat{z}_i)=\log\Big(1+\frac{\lambda_i}{d_i}\Big).$
Thus in bits $r_i(d_i)=\log_2(1+\lambda_i/d_i)$, and additivity yields \eqref{eq:rate_sum}.

The mismatched MMSE gain using assumed variance $\lambda_i^{(b)}$ is
\[
k_i=\frac{\lambda_i^{(b)}}{\lambda_i^{(b)}+d_i}.
\]
Hence $\widetilde{z}_i=k_i\widehat{z}_i=k_i(z_i+q_i)$ and error is
\[
\varepsilon_i=z_i-\widetilde{z}_i=(1-k_i)z_i-k_i q_i.
\]
By independence,
\[
\E[|\varepsilon_i|^2]=|1-k_i|^2\lambda_i+|k_i|^2 d_i.
\]
Substituting $k_i=\lambda_i^{(b)}/(\lambda_i^{(b)}+d_i)$ gives
\[
|1-k_i|^2=\Big(\frac{d_i}{\lambda_i^{(b)}+d_i}\Big)^2,\quad
|k_i|^2=\Big(\frac{\lambda_i^{(b)}}{\lambda_i^{(b)}+d_i}\Big)^2,
\]
and therefore
\[
e_i(d_i)=\E[|\varepsilon_i|^2]
=\frac{\lambda_i d_i^2+(\lambda_i^{(b)})^2 d_i}{(\lambda_i^{(b)}+d_i)^2},
\]
yielding \eqref{eq:D_sum}. We form the Lagrangian for \eqref{eq:opt_problem}:
\[
\mathcal{L}(\{d_i\},\mu)=\sum_i e_i(d_i)+\mu\Big(\sum_i r_i(d_i)-R\Big),
\quad \mu\ge 0.
\]
For interior optima, stationarity requires $e_i'(d_i)+\mu r_i'(d_i)=0$.
Compute
\[
r_i'(d) = \frac{\partial}{\partial d}\log_2\Big(1+\frac{\lambda_i}{d}\Big)
=-\frac{1}{\ln 2}\cdot\frac{\lambda_i}{d(\lambda_i+d)}.
\]
Differentiate $e_i(d)$ to obtain
\[
e_i'(d)=\frac{\lambda_i^{(b)}\big((\lambda_i^{(b)})^2+(2\lambda_i-\lambda_i^{(b)})d\big)}{(\lambda_i^{(b)}+d)^3}.
\]
Setting $e_i'(d_i) = \mu(-r_i'(d_i))$ and clearing denominators yields \eqref{eq:cubic}.
 
If $\lambda_i^{(b)}=\lambda_i$, then $e_i(d)=\lambda_i d/(\lambda_i+d)$ and the KKT condition boils down to classical RWF $d_i^\star=\min(\gamma,\lambda_i)$ \cite{CoverThomas,TC2000TIT}.
\end{proof}

 \section{Numerical Results}\label{sec:sim}
We evaluate the proposed RRWF rule for CSI compression in massive MIMO-OFDM under \emph{decoder-side covariance mismatch}. The UE compresses its MMSE estimate \(\mathbf{\tilde h}_k\sim\mathcal{CN}(\mathbf{0},\Cu)\), and we report the
end-to-end benchmark
\begin{equation}
D_k^{\mathrm{E2E}}(R)=D_k^{\mathrm{mmse}}+D_k^{\mathrm{quant}}(R),
\end{equation}
where \(D_k^{\mathrm{mmse}}\) is given in \eqref{eq:Dmmse} and \(D_k^{\mathrm{quant}}(R)\) is computed using Theorem~\ref{thm:mismatchRD} and Theorem~\ref{cor:modifiedWF}. Results are averaged over multiple independent channel covariance realizations.

\subsection{Massive MIMO-OFDM covariance model}\label{subsec:sim_setup}
We consider a wideband MIMO-OFDM channel with \(M=32\) half-wavelength uniform linear array antennas and
\(N=32\) subcarriers, so \(\mathbf{h}_k\in\mathbb{C}^{MN}\) and
\(\mathbf{h}_k\sim\mathcal{CN}(\mathbf{0},\Ch)\). The covariance \(\Ch\) is generated from a
multipath model with \(L=6\) paths,
\begin{equation}\label{eq:multipath_cov_short}
\Ch=\sum_{\ell=1}^L \gamma_{k,\ell}\big(\mathbf{b}(\tau_{k,\ell})\mathbf{b}^{\mathsf{H}}(\tau_{k,\ell})\big)
\otimes
\big(\mathbf{a}(\theta_{k,\ell})\mathbf{a}^{\mathsf{H}}(\theta_{k,\ell})\big),
\end{equation}
with \(\gamma_{k,\ell}\sim\mathrm{Unif}[0.4,0.8]\), \(\theta_{k,\ell}\sim\mathrm{Unif}[-60^\circ,60^\circ]\),
\(\tau_{k,\ell}\sim\mathrm{Unif}[0,\tau_{\max}]\), and \(\sum_{\ell}\gamma_{k,\ell}=1\).
Unless stated otherwise, we use comb-type i.i.d.\ Gaussian pilots \cite{Caire2025} with training length
\(L^{\rm tr} = 8\), reflecting the practically relevant short-training regime.

\begin{figure}[t]
\centering
\includegraphics[width=0.92\linewidth]{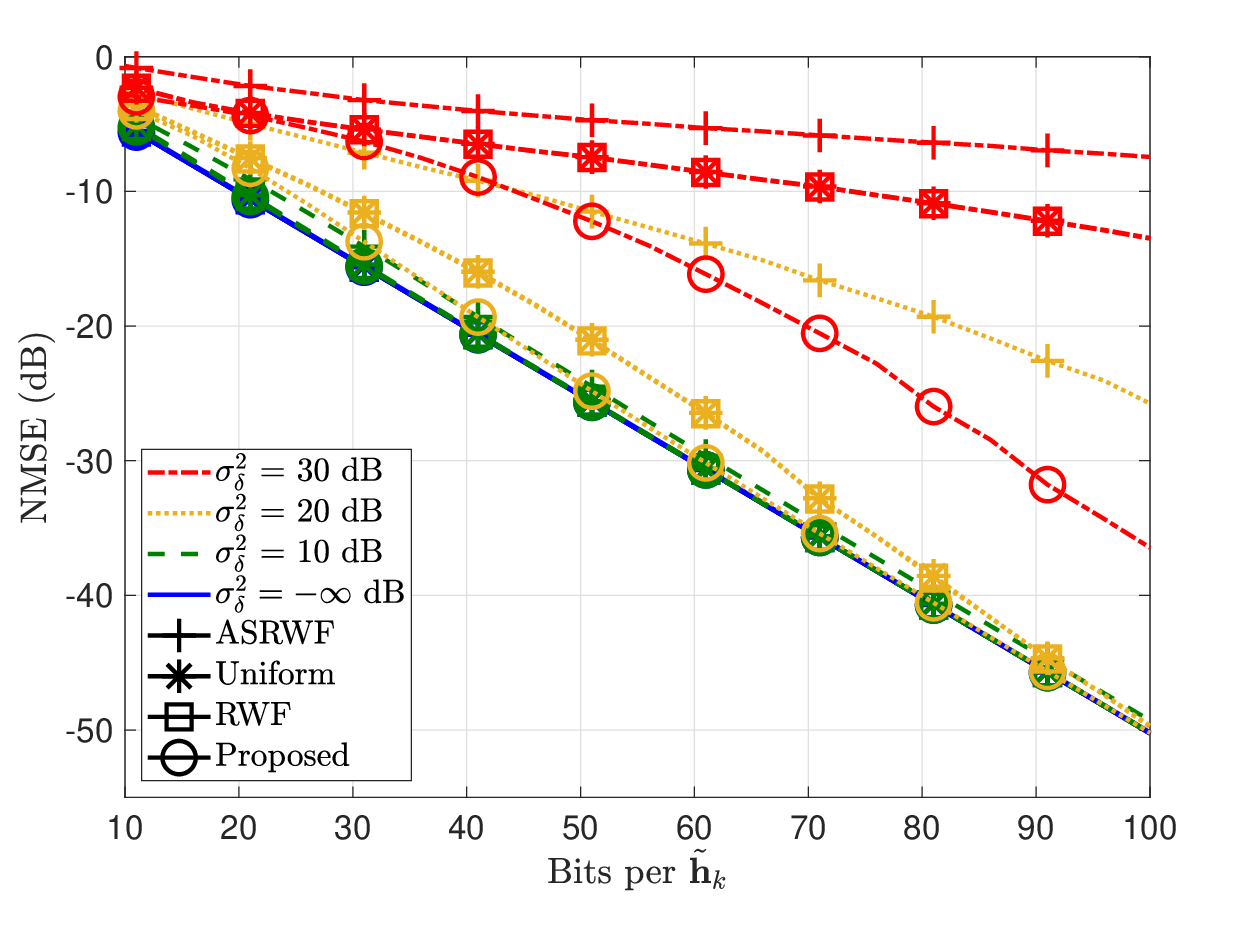}
\vspace{-15pt} 
\caption{NMSE versus feedback rate \(R\) for different mismatch levels \(\sigma_{\delta}\) in \eqref{eq:mismatch_eigs_short}.}
\vspace{-15pt}
\label{fig:dr_curves}
\end{figure}

\subsection{Decoder-Side Mismatch Model}\label{subsec:mismatch_model_sim} We adopt the shared-eigenvector mismatch regime:
\begin{align*}
\Cu&=\Fmat\diag(\lambda_1,\ldots,\lambda_{MN})\Fmat^{\Herm},\\
\Cub&=\Fmat\diag(\lambda^{(b)}_1,\ldots,\lambda^{(b)}_{MN})\Fmat^{\Herm},
\end{align*} with mismatched eigenvalues \begin{equation}\label{eq:mismatch_eigs_short}
\lambda_i^{(b)}=\lambda_i\cdot 10^{\delta_i/10},
\end{equation}
where \(\delta_i\sim\mathcal{N}(0,\sigma_\delta^2)\) (dB). The mismatch level \(\sigma_\delta\)
captures calibration errors, nonstationarity, and covariance learning noise; \(\sigma_\delta=0~(-\infty~\mathrm{dB})\)
corresponds to the matched case.

 \subsection{Comparison Schemes and Performance Metric}
For a target feedback rate \(R\) (bits per realization of \(\mathbf{\tilde h}_k\)), we compare:
\begin{itemize}
\item \textbf{Proposed RRWF:} mismatch-aware allocation from Theorem ~\ref{cor:modifiedWF}, restricted to \(\Cq=\Fmat\diag(d_1,\ldots,d_{MN})\Fmat^{\Herm}\).
\item \textbf{Conventional RWF:} matched RWF using \(\{\lambda_i\}\).
\item \textbf{Assumed-statistics RWF~(ASRWF):} RWF designed on \(\{\lambda_i^{(b)}\}\).
\item \textbf{Uniform allocation:} equal rate over the strongest \(L\) modes (chosen to meet \(R\)).
\end{itemize}
Learning-based methods are not compared because most of them do not exploit shared covariance, and when such statistics are available, simple linear methods can already achieve competitive performance \cite{Caire2025,park2026fundamental}.

All distortions are evaluated under the same mismatched reconstruction functional:
\begin{equation}
D_{\mathrm{quant}}=\sum_{i=1}^{MN} \frac{\lambda_i d_i^2+(\lambda_i^{(b)})^2 d_i}{(\lambda_i^{(b)}+d_i)^2}.
\end{equation}
We report \(\mathrm{NMSE}=D_{\mathrm{quant}}/\tr(\Cu)\) and, when needed,
\(\mathrm{NMSE}_{\mathrm{E2E}}=D_k^{\mathrm{E2E}}(R)/\tr(\Ch)\).

\subsection{Results}\label{subsec:results}
Fig.~\ref{fig:dr_curves} plots NMSE versus feedback rate \(R\) for different mismatch levels.
When \(\sigma_\delta=0\), RRWF reduces to conventional RWF. Under mismatch, RRWF achieves
uniformly lower distortion across the rate range by reallocating bits toward modes that
remain effective under the decoder’s assumed eigenvalues.

Fig.~\ref{fig:mismatch_sweep} fixes \(R\) and varies \(\sigma_\delta\). Conventional RWF
degrades quickly as mismatch grows, while RRWF remains substantially more robust due to
its mismatch-aware KKT allocation.

Finally, Fig.~\ref{fig:e2e_nmse} reports the end-to-end benchmark including the MMSE floor.
RRWF provides the largest gains at low and moderate feedback rates; at high rates, all
schemes approach the MMSE floor, and the gap naturally shrinks.

\begin{figure}[t]
\centering
\includegraphics[width=0.92\linewidth]{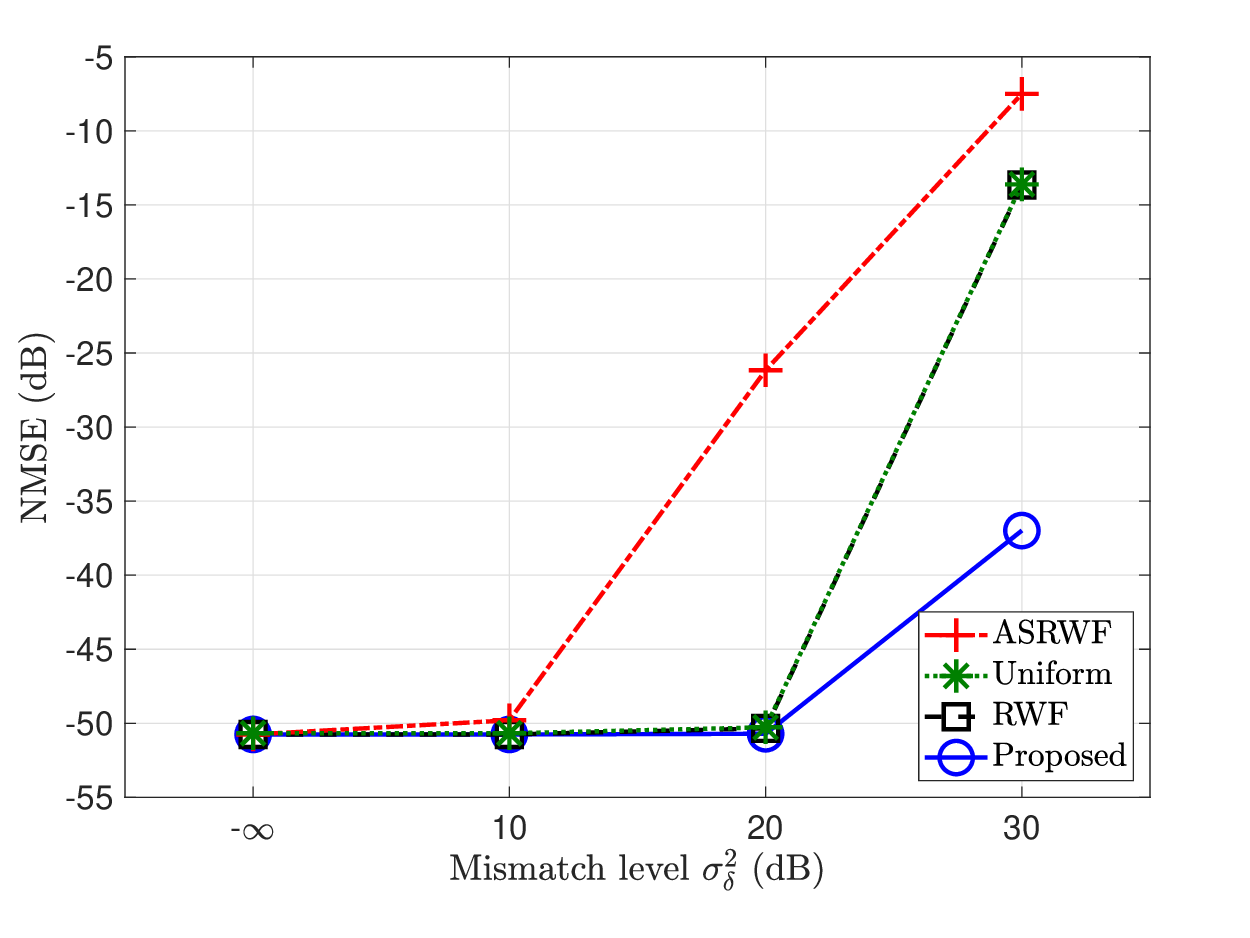}
\vspace{-15pt} 
\caption{NMSE versus mismatch severity \(\sigma_{\delta}\) at a fixed rate (101 bits per \(\mathbf{\tilde h}_k\)).}
\vspace{-15pt}
\label{fig:mismatch_sweep}
\end{figure}

\begin{figure}[t]
\centering
\includegraphics[width=0.92\linewidth]{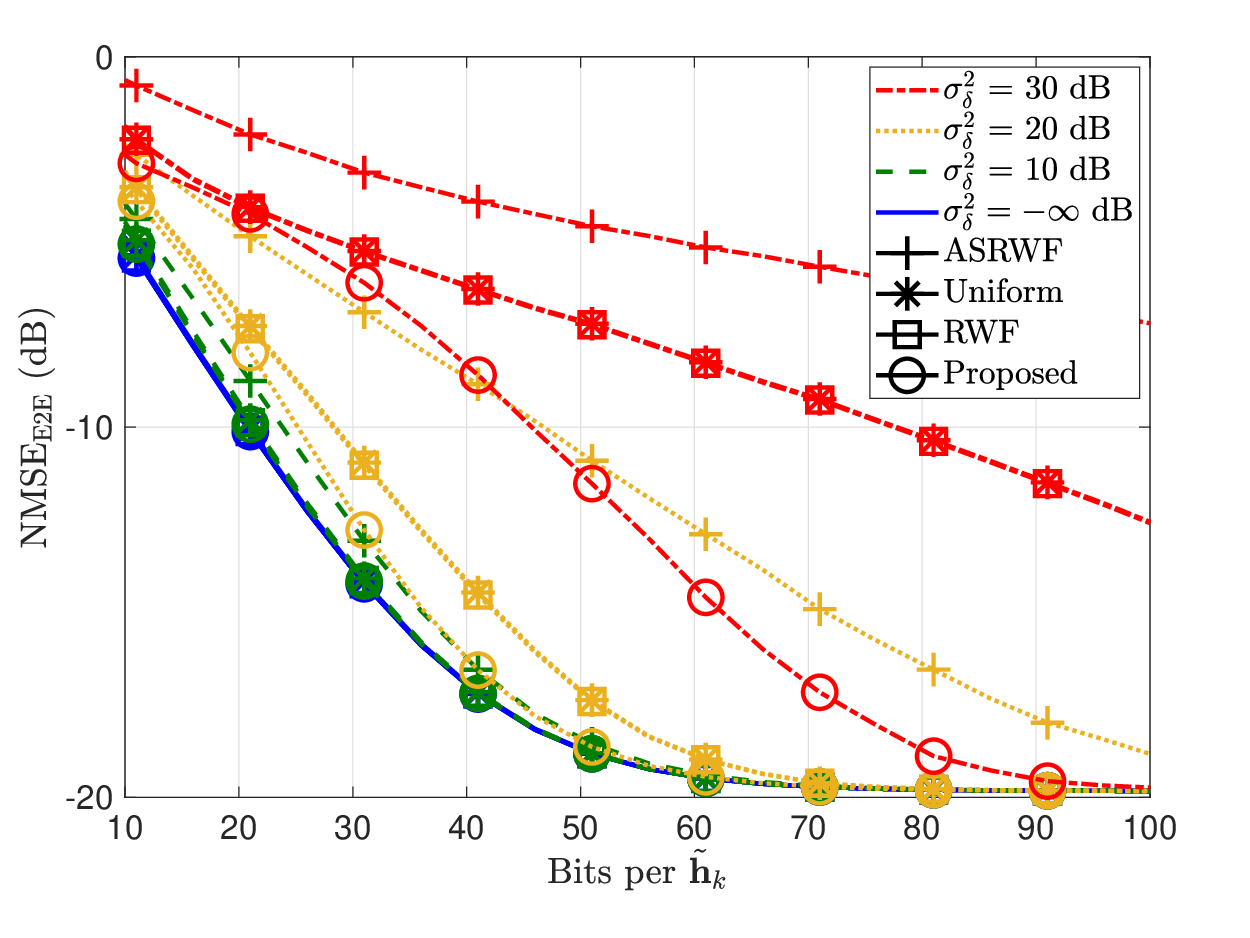}
\vspace{-15pt}  
\caption{End-to-end benchmark \(\mathrm{NMSE}_{\mathrm{E2E}}\) including the MMSE estimation floor \(D_k^{\mathrm{mmse}}/\mathrm{tr}(\mathbf{C}_{\mathbf{h}_k})\approx -20~\mathrm{dB}\).}
\vspace{-15pt}
\label{fig:e2e_nmse}
\end{figure}

\section{Conclusion}
We studied CSI compression for FDD massive MIMO when the BS reconstructs using a mismatched covariance model. We derived an achievable mismatched Gaussian RD benchmark and, in a shared-eigenvector regime, obtained a RRWF allocation with an efficient numerical procedure. Simulations with massive MIMO-OFDM covariance models show that RRWF consistently improves NMSE relative to conventional RWF under decoder-side mismatch, with the largest gains at low and moderate feedback rates. 

\section{Acknowledgment}
This work was supported by the National Research Foundation of Korea (NRF) Grant funded by the Korea government (MSIT) under Grant 2022R1A5A1027646, and by the Korea Research Institute for defense Technology planning and advancement (KRIT) - grant funded by the Defense Acquisition Program Administration (DAPA) (KRIT-CT-22-078).

\bibliographystyle{IEEEtran}
\bibliography{IEEEabrv,bibfile} % 참고문헌
\end{document}